\documentclass[a4paper]{article}
\usepackage[english]{babel}
\usepackage{amscd}
\usepackage{amsmath}
\usepackage{amsthm}
\usepackage{amssymb, faktor}
\usepackage{graphicx}
\usepackage[left=2.5cm,right=2.5cm, top=3cm,bottom=3cm,bindingoffset=0cm]{geometry}
\usepackage{enumerate}

\usepackage[numbers, sort&compress]{natbib}

\theoremstyle{plain}

  \usepackage{thmtools}

\declaretheorem[style=nopar]{lemma}
    \newtheorem{theorem}{Theorem}
    \newtheorem{consequence}{Corollary}[section]
    \newtheorem{statement}{Proposition}[section]

\theoremstyle{remark}

  \theoremstyle{definition}
    
        \newtheorem{cexample}{Counterexample}[section]
            \newtheorem{remark}{Remark}[section]

\newcommand{\Ker}[1]{\mathrm{Ker} \, #1}

\newcommand{\diff}[1]{\mathrm{d}  #1}
\newcommand{\diffFX}[2]{ \dfrac{\partial #1}{\partial #2} }

\newcommand{\diffXp}[1]{ \frac{\mathrm{d} }{\diff #1} }

\newcommand{\R}{\mathbb{R}}
\newcommand{\Complex}{\mathbb{C}}

\newcommand{\CP}{{\mathbb{C}}\mathrm{P}}

\newcommand{\Imm}[1]{\mathrm{Im} \, #1}
\newcommand{\Tr}[1]{\mathrm{Tr} \, #1}

\newcommand{\g}{\mathfrak{g}}

\newcommand{\so}{\mathfrak{so}}

\newcommand{\GL}{\mathrm{GL}}

\newcommand{\U}{\mathrm{U}}
\newcommand{\E}{\mathrm{E}}

\newcommand{\gl}{\mathfrak{gl}}
\newcommand{\un}{\mathfrak{u}}

\renewcommand{\deg}{\mathrm{deg}\,}
\title{Algebraic geometry and stability for integrable systems}
\author{Anton Izosimov\footnote{Moscow State University and Higher School of Economics. E-mail: a.m.izosimov@gmail.com}
}
\date{}
\begin{document}
\maketitle
\abstract{In 1970s, a method was developed for integration of nonlinear equations by means of algebraic geometry. Starting from a Lax representation with spectral parameter, the algebro-geometric method allows to solve the system explicitly in terms of theta functions of Riemann surfaces. However, the explicit formulas obtained in this way fail to answer qualitative questions such as whether a given singular solution is stable or not. In the present paper, the problem of stability for equilibrium points is considered, and it is shown that this problem can also be approached by means of algebraic geometry. }
\section{Introduction}
%In 1970s, a method was developed for integration of nonlinear equations by means of algebraic geometry; see \cite{DKN, DMN, babelon} and references therein. 
As is well-known, many finite-dimensional integrable systems can be explicitly solved by means of algebraic geometry.
% (see \cite{DKN, DMN, babelon} and references therein). 
The starting point for the algebro-geometric integration method is Lax representation.  A dynamical system is said to admit a \textit{Lax representation }\textit{with spectral parameter $\lambda$} if the following two conditions are satisfied.
\begin{enumerate}
\item The phase space of the system can be identified with a certain submanifold $\mathcal L$ of the space $\gl(n,\Complex) \otimes \Complex(\lambda) $ of matrix-valued functions of a complex variable $\lambda$.
\item Under this identification, equations of motion take the form
\begin{align}
\label{LaxRepr1}
\diffXp{t}  L_\lambda = [ L_\lambda,  A_\lambda(L_\lambda)]
\end{align}
where $L_\lambda \in \mathcal L$ is the phase variable, and $ A_\lambda$ is a mapping $A_\lambda \colon \mathcal L \to \gl(n,\Complex) \otimes \Complex(\lambda) $.  
\end{enumerate}
 %The dependence of $ L_\lambda(x)$ and $A_\lambda(x)$ on $\lambda$ is normally assumed to be rational or, more generally, algebraic.
%$X \to \gl(n,\Complex) \otimes \Complex(\lambda)$ 
%which maps the 
%can be embedded into the space of matrix-valued functions   $\gl(n,\Complex) \otimes \Complex(\lambda)$ on which the system takes the form
%
%\begin{align}
%\label{LaxRepr}
%\diffXp{t}  L_\lambda(x) = [ L_\lambda( x),  A_\lambda(x)]
%\end{align}
%where $ L_\lambda(x)$ and $ A_\lambda(x)$ are matrix-valued functions of the phase variable $x \in X$ and the parameter $\lambda \in \Complex$, wherein it is assumed that the mapping $ L_\lambda \colon X \to \gl(n,\Complex) \otimes \Complex(\lambda)$, which sends $x$ to $ L_\lambda(x)$, is injective. The dependence of $ L_\lambda(x)$ and $ A_\lambda(x)$ on $\lambda$ is assumed to be rational or, more generally, algebraic. Since the mapping $x \mapsto  L_\lambda(x)$ is injective, it is convenient to identify $X$ with its image $ L_\lambda(X)$, and to treat $ L_\lambda \in \gl(n,\Complex) \otimes \Complex(\lambda)$ as a phase variable. Regarded in this way, equation \eqref{LaxRepr} can be written simply as
%We will keep the notation $X$ for the phase space of \eqref{LaxRepr1}, i.e. we will identify $X$ and $ L_\lambda(X)$.
\par
%We will still use the notation $X$ for the phase space which \par
%As an example, consider the equations of motion of a torque-free $n$-dimensional rigid body. These equations have the form
% \begin{align}\label{eae}
%\begin{cases}
% \dot M= [M, \Omega],\\
%M = \Omega J + J \Omega
%\end{cases}
% \end{align}
% where
%$M \in \so(n)$ is the angular momentum matrix,
% $\Omega \in \so(n)$ is the angular velocity matrix, and $J$ is a constant positive symmetric matrix called the mass tensor (see \cite{Arnold, Manakov2, nlin} for details).
% As was noted by Manakov \cite{Manakov}, equations \eqref{eae} can be rewritten as
%\begin{align}\label{Manakov}
% \diffXp{t} (M+\lambda J^2) = [M+\lambda J^2, \Omega + \lambda J],
% \end{align}
% so the system admits a Lax representation with spectral parameter.\par
 Starting from a Lax representation with spectral parameter, the algebro-geometric integration method allows to write the solution of the system with initial data $   L_\lambda =  L_\lambda^0$ in terms of theta functions associated with the algebraic curve 
\begin{align}
\label{char}
\det( L_\lambda^0 - \mu \E) = 0,
\end{align}
which is called the \textit{spectral curve}. We refer the reader to \cite{DKN, DMN, Dubrovin, babelon, Audin} and references therein for more details on the algebro-geometric integration method.\par
%is the characteristic polynomial of $ L_\lambda(x)$. 
%Explicit formulas obtained in this way solve (\ref{LaxRepr}) for generic initial data.
Despite the possibility to explicitly solve the equation \eqref{LaxRepr1} in terms of theta functions,  if we are interested in {qualitative} features of dynamics, theta-functional formulas seem to be of little use at least for the following reasons. Firstly, theta-functional solutions correspond to non-singular spectral curves, while most remarkable solutions, such as fixed points or stable periodic solutions, are related to degenerate curves. Secondly, theta-functional formulas provide solutions of the complexified system, and it is in general a difficult problem to describe real solutions. At the same time, many dynamical phenomena, such as stability, are related to the presence of a real structure.
\par

% there are several reasons why the explicit formulas obtained in this way seem to be of little use for understanding qualitative features of dynamics. At the same time, many dynamical phenomena such as stability are related to the presence of a real structure.
%\par
%at least for the following reasons:
%Given a singular solution, such as an equilibrium point or a periodic trajectory, how do we determine whether it is stable or not? 
%One of the reasons why it is hard to answer this question by looking at the explicit formulas is that theta-functional solutions correspond to non-singular spectral curves, while singular solutions correspond to degenerate ones.
% the problem of stability for singular solutions, i.e. equilibrium points, periodic trajectories etc. 
%In particular, theta-functional formulas are usually written for solutions corresponding to non-singular spectral curves, while 
%
% whether a given singular solution of (\ref{LaxRepr}) stable or not?\par \smallskip
%The idea of the present paper is to look at singular spectral curves and 
In the present paper we study the Lyapunov stability problem for systems which admit a Lax representation with spectral parameter. We show that this problem can also be approached by means of classical algebraic geometry, and that this approach is very natural and fruitful. Though we focus on stability of equilibrium points, we expect that our results can  be generalized to more general solutions.
We note that the relation between topology of integrable systems and algebraic geometry has been extensively studied by M.\,Audin and her collaborators \cite{Audin, audin2}, however it seems that their approach cannot be directly used to study the stability problem.\par
Before we formulate the main result of the paper, let us describe the class of Lax matrices which we consider. Firstly, for the sake of simplicity, we restrict ourselves to the case when $ L_\lambda$ is polynomial in parameter $\lambda$, i.e. when $\mathcal L \subset \gl(n,\Complex) \otimes \Complex[\lambda]$. Note that it is more standard to consider Lax matrices which are polynomial in $\lambda$ and $\lambda^{-1}$, i.e. which belong to the loop algebra $ \gl(n,\Complex) \otimes \Complex[\lambda,\lambda^{-1}]$. However this situation can be reduced to the polynomial case by multiplying $L_\lambda$ by a suitable power of $\lambda$.\par
 Our second assumption is the following: there exists an anti-holomorphic involution $\tau \colon  \Complex \to  \Complex$ and a complex number $\sigma$ of absolute value $1$ such that for each $ L_\lambda \in \mathcal L$ \begin{align}
\label{involution}
 L_{\tau(\lambda)} = \sigma L_\lambda^*.
\end{align}
%There are two reasons why this assumption seems to be natural. Firstly, it is satisfied for a wide class of integrable systems. 
%This assumption is satisfied for a wide class of integrable systems. 
The presence of such involution is a common feature of many integrable systems. The most standard examples are provided in Table \ref{table1}. See \cite{Manakov, bobenko1989, Jovanovich, Adler, Adler2} for more details on these systems and their Lax representations. 
%Secondly, as it is said above, stability is a related to the presence of a real structure, i.e. an anti-holomorphic involution.\par
\begin{table}
\centerline{\begin{tabular}{c|c|c|c}\textbf{Integrable system} & \textbf{Lax matrix} &\textbf{} $\tau$ & $\sigma$ 
%& \textbf{Reference}
 \\\hline Euler-Manakov top & {$L_\lambda = A+ B\lambda$}, & $\lambda \mapsto - \bar \lambda\vphantom{\int_a^b}$ & -1 
%& \cite{Manakov}
 \\ & $A^* = -A, B^* = B$ &&  \\\hline Kowalevski top & $L_\lambda = A + B\lambda + C\lambda^2,$ & $\lambda \mapsto - \bar \lambda\vphantom{\int_a^b}$ & 1 
 %& \cite{bobenko1989}
  \\ & $A^*=A, B^* = -B, C^*=C$&& \\\hline Geodesic flow on ellipsoid & $L_\lambda = A + B\lambda + C\lambda^2,$ & $\lambda \mapsto - \bar \lambda\vphantom{\int_a^b}$ & 1 
  %&  \cite{Jovanovich} 
  \\ & $A^*=A, B^* = -B, C^*=C$&& \\\hline Lagrange top & $L_\lambda = A + B\lambda + C\lambda^2$,& $\lambda \mapsto  \bar \lambda\vphantom{\int_a^b}$ & -1% &  \cite{Adler}
   \\&  $A^*=-A, B^* = -B, C^*=-C$&& 
%\\\hline Periodic Toda Lattice & $L_\lambda = A  + B\lambda + A^*\lambda^2$ &$\lambda \mapsto {\bar \lambda}^{-1}\vphantom{\int_a^b}$ & ${\bar \lambda}^{-2}$ %&\cite{Adler2}
%    \\ & $B^* = B$ &&
    \end{tabular}}
\caption{Integrable systems which admit Lax representation with $L_{\tau(\lambda)} = \sigma L_\lambda^*$. }\label{table1}
\end{table}
%For example, Lax Matrices for the Euler \cite{Manakov} and Kowalevski \cite{bobenko1989} admits an involution $\tau(\lambda) = -\bar \lambda$, Lax matrix for the Lagrange top \cite{Adler2, Adler} admits an involution $\tau(\lambda) = \bar \lambda$, and 

%The existence of the involution $\tau$ is closely related to the stability phenomenon. 
Let us now formulate the main result.
Let $  L_\lambda^0 \in \mathcal L$, and consider the associated spectral curve \eqref{char}. The involution $\tau$ induces an anti-holomorphic involution $$\widehat \tau \colon (\lambda, \mu) \mapsto (\tau(\lambda),\sigma\bar \mu)$$ on the spectral curve. We show that if  $  L_\lambda^0$ is a fixed point of  the Lax equation \eqref{LaxRepr1}, then, under some additional assumptions, a sufficient condition for its stability is that all singular points of the associated spectral curve are fixed points of $\widehat \tau$. If we interpret $\widehat \tau$ as a real structure, then this condition means that all singular points lie in the real part of the curve. We consider several examples in which this condition turns out to be necessary and sufficient.\par
Our first example is the Lagrange top. We use the algebro-geometric approach to recover the classical result that the rotation of a top is stable if the angular velocity is greater or equal than some critical value.\par
The second example is the top on a compact Lie algebras defined by a\textit{sectional operator}. This system is related to the so-called \textit{argument shift method}, see Mischenko and Fomenko \cite{MF}.\par
The third, and the most interesting, example is the free multidimensional rigid body, or Euler-Manakov top. It is a standard result that the rotation of a torque-free three-dimensional rigid body about the short or the long axis of inertia is stable, whereas the rotation about the middle axis is unstable. Using the algebro-geometric approach, we obtain a multidimensional generalization of this result. We note that this problem has previously been approached by different methods \cite{Marshall, Spiegler, Casu, Ratiu, JGP, nlin}, however no complete solution has been known.
\section{Stability for integrable and Lax systems}
Let $ \dot x = v(x)$ be a dynamical system on a manifold $X$, and assume that $f_1, \dots, f_N$ are its (in general, complex-valued) first integrals. The \textit{moment map} is a map $F \colon X \to \Complex^N$ which maps $x \in X$ to $(f_1(x), \dots, f_N(x))$.
\begin{statement}\label{stabProp}
Assume that $x_0 \in X$ is an isolated point in the level set of the moment map. Then $x_0$ is Lyapunov stable fixed point of $ \dot x = v(x)$.
\end{statement}
\begin{proof}
Let $x(t)$ be the solution with $x(0) = x_0$. Then $F(x(t)) = F(x(0))$, so $x(t) \in F^{-1}(F(x_0))$. Since $x_0$ is isolated in $F^{-1}(F(x_0))$, this implies that $x(t) = x_0$, i.e. $x_0$ is a fixed point. To prove stability, note that
$$
f(x) = \sum_{i=1}^N |f_{i}(x) - f_{i}(x_0)|^2.
$$
is a Lyapunov function.
\end{proof}
\begin{remark}\label{unstab}
As was shown by Bolsinov and Borisov \cite{BolBorStab}, a similar statement is true for periodic trajectories: if a periodic trajectory coincides with a connected component of the level set of the moment map, then it is stable. Moreover, under some additional assumptions, the converse is also true. In our case, the following is true. Let $X$ be Poisson manifold, and let $ \dot x = v(x)$ be a Hamiltonian system. Assume that $f_1, \dots, f_N$ is a complete family of analytic first integrals in involution. Further, assume that the level sets of $F$ are compact, so that their connected components are invariant tori, and that $ \dot x = v(x)$ is a non-resonant system, which means that its trajectories are dense on almost all tori \cite{intsys}. Then the condition of Proposition \ref{stabProp} is necessary and sufficient. 

\end{remark}
Now, let us reformulate Proposition \ref{stabProp} for systems which admit a Lax representation with spectral parameter. Consider the space
 $$\mathcal P_{m,n} =\{ L_\lambda = B_m\lambda^m + \dots +  B_0 \in \gl(n,\Complex) \otimes \Complex[\lambda] \}$$ of $\gl(n,\Complex)$-valued polynomials of degree $m$, and let $\mathcal L \subset \mathcal P_{m,n}$ be its submanifold. Let $A_\lambda$ be a map $A_\lambda \colon \mathcal L \to \gl(n,\Complex) \otimes \Complex(\lambda) $, and assume that $\mathcal L$ is invariant with respect to the flow
\begin{align}\label{LaxRepr2}
\diffXp{t}  L_\lambda = [ L_\lambda,  A_\lambda(L_\lambda)].
\end{align}
To each $L_\lambda \in \mathcal L$ we can assign its {spectral curve}, i.e. an affine algebraic curve $C(L_\lambda)$ given by the equation
$
P(\lambda, \mu) =  0
$
where
$$
P(\lambda,\mu) = \det( L_\lambda - \mu \E)
$$
is the characteristic polynomial of $ L_\lambda$. 
The following is well-known.
\begin{statement}\label{curvInt}
Let $L_\lambda(t)$ be a solution of \eqref{LaxRepr2}. Then the curve $C(L_\lambda(t))$ does not depend on $t$.
\end{statement}
\begin{proof}
Equation (\ref{LaxRepr2}) implies that
\begin{align*}
%\label{LaxRepr2}
\diffXp{t} ( L_\lambda)^k = [( L_\lambda)^k,  A_\lambda]
\end{align*}
for any positive integer $k$. Therefore, the function $\Tr  ( L_\lambda)^k$ is an integral of motion for any values of $k$ and $\lambda$, and so are the coefficients of the characteristic polynomial $P(\lambda, \mu)$.
\end{proof}
\begin{statement}\label{LaxStab}
Let $L_\lambda^0 \in \mathcal L$, and assume that $L_\lambda^0$ is an isolated point in the isospectral variety $$S(L_\lambda^0) = \{ L_\lambda \in \mathcal L \mid C(L_\lambda) = C(L_\lambda^0)\}.$$
Then $L_\lambda^0$ is a stable fixed point of (\ref{LaxRepr2}).
\end{statement}
\begin{proof}
This follows from Proposition \ref{stabProp} and Proposition \ref{curvInt}.
\end{proof}
In the next section, we present algebro-geometric conditions which imply the hypothesis of Proposition \ref{LaxStab}.
% on a phase space $\mathcal L$, and assume $L_\lambda$ is a polynomial in $\lambda$ of a given degree $m$, i.e. that $\mathcal L \subset \gl(n,\Complex) \otimes \Complex[\lambda]$.
\section{A lemma on polynomial matrix pencils and stability theorem}\label{agLemma}
%Let $ L_\lambda =  B_m\lambda^m + \dots +  B_0 \in \gl(n,\Complex) \otimes \Complex[\lambda]$ be a $n \times n$ matrix polynomially depending on the parameter $\lambda$. The \textit{spectral curve} associated with $ L_\lambda$ is the affine algebraic curve $C \subset \Complex^2$ given by
%$
%P(\lambda, \mu) =  0
%$
%where
%$
%P(\lambda,\mu) = \det( L_\lambda - \mu \E)
%$
%is the characteristic polynomial of $ L_\lambda$. 
Let $B \in  \gl(n,\Complex)$ be a fixed $n \times n$ matrix, and let $$\mathcal P_{m,n}(B) =\{ L_\lambda = B_m\lambda^m + \dots +  B_0 \in \gl(n,\Complex) \otimes \Complex[\lambda] \mid B_m = B  \}$$
be the set of $ \gl(n,\Complex)$-valued polynomials of degree $m$ with leading coefficient $B$. 
%Let $\tau \colon \Complex \to \Complex$ be an anti-holomorphic involution, and let 
%$$X_\tau =\{ L_\lambda = B_m\lambda^m + \dots +  B_0 \in \gl(n,\Complex) \otimes \Complex[\lambda] \mid  L_{\tau(\lambda)} = \pm L_\lambda^*  \}.$$
To each matrix polynomial $L_\lambda \in \mathcal P_{m,n}(B)$ we assign its {spectral curve}, i.e. an affine algebraic curve $C(L_\lambda)$ given by the equation
$
P(\lambda, \mu) =  0
$
where
$$
P(\lambda,\mu) = \det( L_\lambda - \mu \E).
$$The space $\mathcal P_{m,n}(B)$ carries an action of the {gauge group} $$G  = {\{ Q \in \GL(n,\Complex) \mid [Q,B] = 0\}}\big/{\{\nu E \mid \nu \in \Complex^*\}}.$$ This action is given by $L_\lambda \mapsto Q^{-1}L_\lambda Q$, and for each polynomial $L_\lambda \in \mathcal P_m(B)$ we have
$$
 C(L_\lambda) = C(Q^{-1}L_\lambda Q), 
$$
i.e. the spectral curve is the same for matrix polynomials belonging to the same gauge group orbit. However, the converse is not true: two matrix polynomials which have the same spectral curve do not necessarily belong to the same gauge group orbit. 
\begin{remark}
More precisely, let us consider the variety of matrix polynomials isospectral with $L_\lambda$, i.e. the set
$$
S(L_\lambda) = \{ M_\lambda \in \mathcal P_{m,n}(B) \mid C(M_\lambda) = C(L_\lambda)\}.
$$
%Then, if we assume that the spectral curve of $L_\lambda$ is non-singular, and that $B$ has simple spectrum, then the set of matrix polynomials isospectral with $L_\lambda$, i.e. the set
%$$
%S(L_\lambda) = \{ M_\lambda \in \mathcal P_{m,n}(B) \mid C(M_\lambda) = C(L_\lambda)\}
%$$
Then, if we assume that the spectral curve of $L_\lambda$ is non-singular, and that $B$ has simple spectrum, the quotient $S(L_\lambda)/{G}$ can be identified with a Zariski open subset of the Jacobian of the spectral curve, see van Moerbeke and Mumford \cite{mvm}. In the case when the spectral curve is singular, the  description of the quotient is not that transparent, however it is true that $S(L_\lambda) / G$ contains a Zariski open subset which can be identified with a Zariski open subset of the generalized Jacobian of the spectral curve  \cite{mvm}. In particular, the dimension of the isospectral variety is still much bigger than the dimension of the $G$-orbit. 
\end{remark}

 The situation changes drastically if we restrict our attention to matrix-valued polynomials which admit an anti-holomorphic involution.
Let $\tau \colon \Complex \to \Complex$ be an anti-holomorphic involution, and let $\sigma \in \mathrm S^1 = \{z \in \Complex \mid |z|=1\}$. Define
$$\widetilde{\mathcal P}_{m,n}(B, \tau, \sigma) =\{  L_\lambda = B_m\lambda^m + \dots +  B_0  \in \gl(n,\Complex) \otimes \Complex[\lambda]   \mid  L_{\tau(\lambda)} = \sigma L_\lambda^*, \, B_m = B \}.$$
In this case, the gauge group is $$G  = \{ Q \in \U(n) \mid [Q,B] = 0\} \big/  \{\nu E \mid \nu \in \mathrm{S}^1\}.$$For each $L_\lambda \in \widetilde{\mathcal P}_{m,n}(B, \tau, \sigma) $, there is an anti-holomorphic involution $$\widehat \tau \colon (\lambda, \mu) \mapsto (\tau(\lambda), \sigma\overline \mu)$$ on its spectral curve. Denote by $\mathrm{Fix}\,\widehat \tau$ the fixed points set of this involution. If we interpret $\widehat \tau$ as a real structure, then $\mathrm{Fix}\,\widehat \tau$ is the set of real points on the curve. 
%Note that $\mathrm{Fix}\,\widehat \tau = \{ (\lambda, \mu) \in C(L_\lambda) \mid \tau(\lambda) = \lambda \}$.
%To each matrix polynomial $L_\lambda \in X_\tau$ we can assign its \textit{spectral curve}, i.e. an affine algebraic curve $C(L_\lambda)$ given by
%$
%P(\lambda, \mu) =  0
%$
%where
%$
%P(\lambda,\mu) = \det( L_\lambda - \mu \E)
%$
%is the characteristic polynomial of $ L_\lambda$. \par
%The set $X_\tau$ carries the action of the unitary group $\U(n)$ given by $L_\lambda \mapsto Q^{-1}L_\lambda Q$, and for each polynomial $L_\lambda \in X_\tau$ we have
%$$
%C(L_\lambda) = C(Q^{-1}L_\lambda Q), 
%$$
%i.e. the spectral curve is the same for matrix polynomials belonging to the same $\U(n)$ orbit. In general, the converse is not true: two matrix polynomials which have the same spectral curve do not necessarily belong to the same $\U(n)$ orbit.
%
\begin{lemma}\label{lemma11}
Let $L_\lambda^0 \in \widetilde{\mathcal P}_{m,n}(B, \tau, \sigma) $, and assume that the leading term $B$ has simple spectrum. Let $C = C(L_\lambda^0)$ be the spectral curve associated with $L_\lambda^0$. Assume that
\begin{enumerate}
%\item each singular point of $\Gamma$ belongs to $\mathrm{Fix}\,\widehat \tau$.
{\item all singular points of $C$ lie in $ \mathrm{Fix}\,\widehat \tau$;}
\item each irreducible component of $C$ is smooth and has genus zero;
%\item each $\Gamma_i$ is either smooth, or all its singular points are ordinary multiple points;
\item each intersection between two distinct irreducible components is at worst of order $2$.
\end{enumerate}
Then the isospectral variety 
$$
S = \{ L_\lambda \in \mathcal P_{m,n}(B) \mid C(L_\lambda) = C\}.
$$
 coincides with the gauge group orbit of $L_\lambda^0$, i.e. if $L_\lambda \in\widetilde{\mathcal P}_{m,n}(B, \tau, \sigma) $, and $C(L_\lambda) = C$, then there exists $Q \in G$ such that
$${ L_\lambda} =  Q^{-1} L_\lambda^0  Q.$$
\end{lemma}
The proof is given in Section \ref{proofSec}.
\begin{remark}\label{weak}
We note that since each irreducible component of $C$ is smooth, the singularities of $C$ are exactly those points which belong to at least two components.
At the same time, the smoothness condition can actually be avoided, however if irreducible components of $C$ have self-intersections, they should also satisfy Conditions 1 and 3 of the lemma. We require smoothness not to go deep into singularity theory of algebraic curves. In all examples that we consider, irreducible components of the spectral curve are indeed smooth.\par
Condition 3 can also apparently be weakened. However, in all our examples, irreducible components of the spectral curve are either lines or quadrics, so this condition is automatically satisfied.\par
%\end{remark}
%\begin{remark}
The two remaining conditions, namely Condition 1 and condition on the genus, are crucial. In Section \ref{ceSect} we consider several counterexamples which show that Lemma \ref{lemma11} does not in general hold if one of these two conditions is not satisfied. However, Lemma \ref{lemma11} admits the following generalization on higher genus curves: if all conditions of Lemma \ref{lemma11} are satisfied except, possibly, the condition on the genus, then the quotient of the isospectral variety by the action of the gauge group can be identified with a subset of the Jacobian of $C$. Note that, in the genus zero case, this is exactly Lemma \ref{lemma11}, since the Jacobian is a single point. Also note that if Condition 1 of Lemma \ref{lemma11} is violated, then the Jacobian should be replaced by the generalized Jacobian. %See also Counterexample \ref{ce1}.
%Also note that if condition 1 of Lemma \ref{lemma11} is not satisfied, then the assertion of Proposition \ref{genLemma} is false. In this case, it can be shown that the quotient of the isospectral variety by the action of the gauge group is a subset of the \textit{generalized} Jacobian, or, more precisely, its compactification.
\end{remark}

Now we are in a position to state the main result of the paper.
\begin{theorem}\label{thm1}
Consider a system which admits a Lax representation with spectral parameter such that its phase space $\mathcal L$ lies in $\widetilde{\mathcal P}_{m,n}(B, \tau, \sigma)$. Assume that 
\begin{enumerate}
\item $L_\lambda^0 \in \mathcal L$ satisfies conditions of Lemma \ref{lemma11};
\item the intersection of the gauge group orbit of $L_\lambda^0$ with $\mathcal L$ is a discrete set.
\end{enumerate} Then $L_\lambda^0$ is a Lyapunov stable fixed point.
\end{theorem}
\newpage
\begin{proof}
Lemma \ref{lemma11} implies that the isospectral variety $$S(L_\lambda^0) = \{ L_\lambda \in \widetilde{\mathcal P}_{m,n}(B, \tau, \sigma) \mid C(L_\lambda) = C(L_\lambda^0)\}$$ coincides with $G$-orbit of $L_\lambda^0$, so 
$$ \{ L_\lambda \in \mathcal L \mid C(L_\lambda) = C(L_\lambda^0)\} = S(L_\lambda^0) \cap \mathcal L$$
is a discrete subset of $\mathcal L$. Now apply Proposition \ref{LaxStab}.
\end{proof}
\begin{remark}
\label{common}
Note that the Lax flow \eqref{LaxRepr2} can be, as a rule, included into an hierarchy of commuting Lax flows
\begin{align}
\label{higherLax}
 \diffXp{t}  L_\lambda = [L_\lambda, A_\lambda^i(L_\lambda)], \quad i =1,\dots, k.
\end{align}
Theorem \ref{thm1} is applicable for all these flows. In particular, if $L_\lambda^0 \in \mathcal L$ satisfies the assumptions of the theorem, then it is a common equilibrium point for all flows \eqref{higherLax}. Note that if $L_\lambda^0 \in \mathcal L$ is an equilibrium point of  \eqref{LaxRepr2}, but is not a common equilibrium point, then $L_\lambda^0$ is automatically not isolated in the moment map level set, and under some additional assumptions, it is unstable (see Remark \ref{unstab}).
\end{remark}
In Sections \ref{LSect}, \ref{MSect}, \ref{MRBSect} we consider three examples illustrating Theorem \ref{thm1}. In Section \ref{ceSect} we consider several counterexamples showing that assumptions of Lemma \ref{lemma11} or Theorem \ref{thm1} in general cannot be avoided (except those mentioned in Remark \ref{weak}).

\section{Example I: Lagrange top}\label{LSect}
The equations of the Lagrange top read
\begin{align*}
%\label{Lagrange}
\begin{cases}
\dot M &= [M, \Omega] + [\Gamma, \chi], \\
\dot \Gamma &= [\Gamma, \Omega],
\end{cases}
\end{align*}
where $M, \Gamma \in \so(3, \R)$ are dynamical variables,  $\Omega \in \so(3,\R)$ is defined by the relation $M = \Omega J + J\Omega$,
\begin{align*}
\chi = \left(\begin{array}{ccc}0 & 0 & 0 \\0 & 0 & 1 \\0 & -1 & 0\end{array}\right),  \quad  J = \mathrm{diag}(a,b,b)
\end{align*}
where $a,b \in \R$ are the inertia moments of the top. These equations can be rewritten as a single Lax equation with parameter:
\begin{align*}
\diffXp{t}(\lambda^2(a+b)\chi + \lambda M + \Gamma) = [\lambda^2(a+b)\chi + \lambda M + \Gamma, \lambda \chi + \Omega].
\end{align*}
 See \cite{Adler, ratiu1982lagrange, gavrilov} for a detailed discussion of this system and underlying algebraic geometry.\par\smallskip
%\begin{align*}
%J = \left(\begin{array}{ccc}a & 0 & 0 \\0 & b & 0 \\0 & 0 & b\end{array}\right).
%\end{align*}
%see  \cite{Adler, ratiu1982lagrange}. 
Note that the Lax matrix $L_\lambda = \lambda^2(a+b)\chi + \lambda M + \Gamma$ satisfies  $L_{\tau(\lambda)} = - L_\lambda^*$ with $\tau$ given by $\lambda \mapsto \bar \lambda$, which allows the application of Theorem \ref{thm1}.
Consider the fixed point $(M_0, \Gamma_0)$ where $\Gamma_0 = \chi$ and $M_0 =  m\chi$. From the point of view of mechanics, this point corresponds to a rotation about the vertical axis, or the so-called sleeping top. It is a classical result that the sleeping top with $m^2 \geq 4(a+b)$ is stable. Let us recover this result by means of the spectral curve. The equation of the spectral curve associated with $L_\lambda^0 =\lambda^2(a+b)\chi + \lambda M_0 + \Gamma_0$ is
\begin{align*}
P(\lambda, \mu) = -\mu(\mu + \mathfrak{i}((a+b)\lambda^2 + m\lambda+1)^2)(\mu - \mathfrak{i}((a+b)\lambda^2 + m\lambda+1)^2) = 0.
\end{align*}
This curve consists of three irreducible components $C_1, C_2, C_3$ of genus $0$. Since the involution $\hat \tau$ is $(\lambda, \mu) \mapsto (\bar \lambda, -\bar \mu)$, we sketch the spectral curve in the axes $\lambda, \mu \mathfrak{i}$. 
\begin{figure}[t]
{\begin{picture}(500,120)
%\put(0,0){\includegraphics[scale = 0.38]{3d_short.jpg}}
%\put(140,0){\includegraphics[scale = 0.38]{3d_inter.jpg}}
%\put(280,0){\includegraphics[scale = 0.38]{3d_long.jpg}}
%\put(35,23){\tiny{$\lambda_1^2$}}
%\put(61,23){\tiny{$\lambda_2^2$}}
%\put(88,23){\tiny{$\lambda_3^2$}}
\put(50,15){
\qbezier(10,80)(50,-10)(90,80)
\qbezier(10,20)(50,110)(90,20)
\put(0,50){\line(1,0){100}}
%\put(75,0){\line(0,1){80}}
\put(0,-15){(a)\, $m^2 - 4(a+b) > 0$}
%\put(57,24){\tiny{$\lambda_2^2$}}
%\put(76,24){\tiny{$\lambda_3^2$}}
}
\put(180,15){
\qbezier(10,95)(50,5)(90,95)
\qbezier(10,5)(50,95)(90,5)
\put(0,50){\line(1,0){100}}
\put(0,-15){(b)\, $m^2 - 4(a+b) = 0$}
}
\put(310,15){
\qbezier(10,100)(50,10)(90,100)
\qbezier(10,0)(50,90)(90,0)
\put(0,50){\line(1,0){100}}
\put(0,-15){(c)\, $m^2 - 4(a+b) < 0$}
}

\end{picture}}
\caption{Spectral curves for the sleeping top.}\label{LagrangeCurve}
\end{figure}
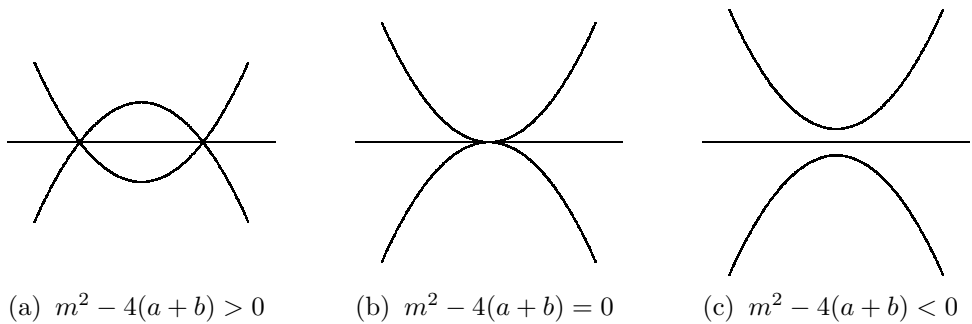Three sketches in Figure \ref{LagrangeCurve} correspond to the cases $D = m^2 - 4(a+b) > 0$, $D = 0$, and $D < 0$.
We conclude that the conditions of Lemma \ref{lemma11} hold if and only if $D \geq 0$, i.e. if the rotation is sufficiently fast.\par Now, to prove that the sleeping top with $D \geq 0$ is stable, it suffices to find the intersection of the gauge group orbit with the phase space. %Find the intersection of the $\U(n)$ orbit of $ L_\lambda(M_0, \Gamma_0) = \lambda^2(a+b)\chi + \lambda M_0 + \Gamma_0$ with the set of possible Lax matrices.
 Let $Q \in G$, and let
$$
L_\lambda = \lambda^2(a+b)\chi + \lambda M + \Gamma = Q( \lambda^2(a+b)\chi + \lambda M_0 + \Gamma_0)  Q^{-1}.
$$
We have $ Q\chi Q^{-1} = \chi$, so $M =  QM_0 Q^{-1} = m Q\chi Q^{-1} = m\chi = M_0$. Analogously, $\Gamma=\Gamma_0$, so the $G$-orbit of $ L_\lambda^0$ is one point, and the equilibrium $(M_0, \Gamma_0)$ is stable.\par\smallskip

%\begin{align}\label{InvLagrange}
%\det(\lambda^2(a+b)\chi + \lambda M + \Gamma - \mu\E) = P(\lambda, \mu)
%\end{align}
%is unique. Since the degrees of the projections of $C_1, C_2, C_3$ to the $\lambda$-plane are equal to one, conclude that $M$ and $\Gamma$ satisfying (\ref{InvLagrange}) must be diagonal in the eigenbasis for $\chi$, which immediately implies that $\Gamma = \chi, M = m\chi$, q.e.d.\par\smallskip

The condition $D \geq 0$ is actually necessary and sufficient for stability \cite{Arnold}. The case when $D=0$ is more complicated compared to $D>0$. It corresponds to the so-called \textit{supercritical Hamiltonian Hopf bifurcation} \cite{Cushman}. 
%Note that it is not possible to apply Theorem \ref{thm1} for $D=0$ because of the tangency point on the spectral curve. However, Theorem \ref{thm2} is still applicable.
%$$
%
%$$
% (the sleeping top \cite{Arnold}).
%and $\Omega \in \so(3, \R)$ is related to $M$ by $M = J\Omega + \Omega J$ where
%$$
%J = \left(\begin{array}{ccc}a & 0 & 0 \\0 & 1-a & 0 \\0 & 0 & 1-a\end{array}\right).
%$$
\section{Example II: Mischenko-Fomenko tops on compact Lie algebras}\label{MSect}
Let $\g$ be a compact simple Lie algebra. The equations of the Mischenko-Fomenko top \cite{MF} are
\begin{align}\label{MFT}
\dot X = [X,\phi(X)]
\end{align}
where $X \in \g$, and $\phi\colon \g \to \g$ is a linear operator satisfying
\begin{align}\label{sectOp}
[\phi(X), A] = [X,B]
\end{align}
for some regular $A \in \g$, and $B \in \mathfrak{C}(a)$ where $\mathfrak C(a) = \{X\in \g \mid [X,A] = 0\}$ is the centralizer of $A$. Operators satisfying (\ref{sectOp}) are called \textit{sectional operators}. The Lax representation with parameter reads\begin{align*}
%\label{MFTL}
\diffXp{t} (X+\lambda A) = [X+\lambda A,\phi(X) - \lambda B].
\end{align*} %Note that this system, written in the form \eqref{MFTL}, can be viewed as a particular case of a general $r$-matrix construction applied to the loop algebra $\g \otimes \R[\lambda, \lambda^{-1}]$, see \cite{reiman1982, semenov} for details.
Note that the Lax matrix $L_\lambda = X + \lambda A$ is a $\g$-valued polynomial. To obtain a matrix-valued polynomial, we pass to any unitary representation $\rho \colon \g \to \un(n)$. The Lax matrix so obtained satisfies  $L_{\tau(\lambda)} = - L_\lambda^*$ with $\tau$ given by $\lambda \mapsto \bar \lambda$.
%Note that in order to apply Theorem \ref{thm2}, the matrix $\rho(A)$ should have simple spectrum. For an arbitrary representation $\rho$, regularity of $A$ does not imply simplicity of the spectrum of $\rho(A)$.
%Note that does not
\begin{statement}\label{MFFixed}
Let $X_0 \in \mathfrak{C}(A)$. Then $X_0$ is a fixed point of (\ref{MFT}). 
\end{statement}
\begin{proof}
Let $X_0 \in \mathfrak{C}(A)$. Since $A$ is regular, $\mathfrak{C}(A)$ is Abelian, and since $B \in \mathfrak{C}(A)$, we have $[X_0,B] = 0$. Using (\ref{sectOp}), we conclude that $\phi(X_0) \in \mathfrak{C}(A)$, so $[X_0, \phi(X_0)] = 0$, q.e.d.
\end{proof}
Let us use Theorem \ref{thm1} to prove that all these equilibria are stable. Note that though $A$ is regular, the spectrum of its image under the representation $\rho$ may be not simple. For example, a matrix $A \in \so(2n, \R)$ with two zero eigenvalues is a regular element. However, it is easy to see that if $\g$ is of type $A_n,B_n,C_n,$ or $G_2$, and $\rho$ is the representation {of minimal dimension}, then regularity of $A$ does imply simplicity of the spectrum of $\rho(A)$. See Konyaev \cite{Konyaev} for details.\par
So, let $A$ and $X_0$ be represented by $n\times n$ matrices. Since $X_0 \in \mathfrak{C}(A)$, the matrices $X_0$ and $A$ are simultaneously diagonalizable. The equation of the spectral curve is
\begin{align*}
P(\lambda, \mu) = \prod_{i=1}^n (x_i + \lambda a_i  - \mu) = 0
\end{align*}
where $\{a_i\} \subset \mathfrak{i}\R$ and $\{x_i\} \subset \mathfrak{i}\R$ are eigenvalues of $A$ and $X_0$ respectively. So, the spectral curve is the union of distinct straight lines 
$
\mu = a_i \lambda + x_i.
$
It is clear that all conditions of Lemma \ref{lemma11} are satisfied, and the gauge group orbit of $L_\lambda^0 = X_0 + \lambda A$ consists of one point. Therefore, $X_0$ is stable.\par
Note that it is possible to generalize Lemma \ref{lemma11} in such a way that $D_n,E_6,E_7,E_8,$ and $F_4$ can also be included. The stability conclusion remains true for these Lie algebras as well. This can also be proved using the technique of \cite{JGP}.
 \begin{remark}
% Note that  \eqref{MFT} may have equilibrium points different from those described \ref{MFFixed}.
We note that the equation \eqref{MFT} is completely integrable can be included into an hierarchy of commuting flows. Equilibrium points described in Proposition \ref{MFFixed} are exactly those which are common for all these flows, see Brailov \cite{Brailov} and Bolsinov and Oshemkov \cite{biham}.
The equation  \eqref{MFT} may have other equilibria, however their stability cannot be studied by the method of the present paper (see Remark  \ref{common}). Nevertheless, if we are able to prove that \eqref{MFT} is a non-resonant system, then we can assert that these equilibria are unstable (Remark \ref{common}). Non-resonance of  \eqref{MFT} for generic $B$ can be deduced from the results of Rybnikov \cite{Rybnikov}.
 \end{remark}
 %so Condition 1 of Theorem \ref{thm2} is not satisfied. 
 \section{Example III: Multidimensional rigid body}\label{MRBSect}
 The equations of motion of a torque-free multidimensional rigid body, also known as the Euler-Manakov top, are
 \begin{align}\label{MRB}
 \dot M = [M,\Omega],
 \end{align}
 $M \in \so(n, \R)$ is a dynamical variable, and $\Omega$ is found from the equation $M = J\Omega + \Omega J$ where $J$ is a fixed symmetric matrix.\par
 
 This example is the most interesting from the point of view of stability. A complete solution of the stability problem is only known for the three-dimensional body. In three dimensions, the equilibria of (\ref{MRB}) are rotations about principal axes, and it is well known that the rotation about the short and the long axis is stable, while the rotation about the middle axis is unstable. For a detailed discussion of the stability problem for the multidimensional rigid body, see \cite{nlin}.
 %The multidimensional generalization of this problem is discussed in \cite{Marshall, Spiegler, Casu, Ratiu, JGP, nlin}, however only partial results are available. 
In this section we show that the solution of this problem by means of the spectral curve is almost straightforward, at least for the so-called \textit{regular} equilibria which are defined below (see also Remark \ref{exotic}).\par
 Let $M_0$ be an equilibrium point of (\ref{MRB}). Then $M_0$ is called regular \cite{NRE} if it is possible to bring $J$ and $M_0$ to the canonical form simultaneously, i.e. if there exists a basis where $J$ is diagonal and $M_0$ takes the form
 	\begin{align*}
M_0= \left(\begin{array}{ccccccc}0 & m_1 &  &  &  &  &    \\-m_1 & 0 &  &  &  &  &    \\ &  & \ddots &  &  &  &    \\ &  &  & 0 & m_l &  &    \\ &  &  & -m_l & 0 &  &    \\ &  &  &  &  & 0 &    \\ &  &  &  &  &  & \ddots  \end{array}\right).
\end{align*}
  Regular equilibria are a natural multidimensional generalization of rotations about principal axes. A complete solution of the stability problem for regular equilibria under the condition that $J$ is generic is given below. \par
 A Lax representation with parameter for the system \eqref{MRB} was found by Manakov \cite{Manakov}. It reads:
 \begin{align}
 \label{laxMRB}
 \diffXp{t}(M+\lambda J^2) = [M+\lambda J^2, \Omega + \lambda J].
 \end{align}
The involution $\tau$ is $\lambda \mapsto -\bar \lambda$. \par
Let $M_0$ be a regular equilibrium point, and let $J^2 = \mathrm{diag}(a_1, \dots, a_n)$
%\begin{align}
%J^2 = \left(\begin{array}{ccc}a_1 &   &   \\ & \ddots &  \\ &  & a_n\end{array}\right)
%\end{align}
where the numbers $a_1, \dots, a_n$ are all distinct.
Then the spectral curve is given by
\begin{align*}
P(\lambda,\mu) = \prod_{i=1}^{l}\left(m_i^2 + ( a_{2i-1}\lambda- \mu)(a_{2i}\lambda - \mu)\right)\prod_{i=2l+1}^n(a_i\lambda - \mu) = 0,
\end{align*}
thus it is a union of hyperbolas and straight lines.
To investigate when the conditions of Lemma \ref{lemma11} are satisfied, we make a change of variables \begin{align*}\begin{cases}x = \dfrac{\mu}{\lambda},\\ y = -\dfrac{1}{\lambda^2}\end{cases}\end{align*}and obtain the curves
\begin{align*}
\begin{aligned}
y = \frac{ (x - a_{2i-1})(x -a_{2i}) }{m_i^2}, \quad &i=1,\dots,l.\\
x = a_i, \quad i = 2l+1&, \dots, n.
\end{aligned}
\end{align*}

The union of these curves is called the \textit{parabolic diagram} associated with $M_0$. 
%Parabolic diagrams were introduced in \cite{JGP} to describe the spectrum of the bi-Hamiltonian structure related to (\ref{MRB}). In the context of spectral curves, the parabolic diagram is the quotient of the spectral curve by the holomorphic involution $(\lambda, \mu) \to (-\lambda, -\mu)$. The relation between the bi-Hamiltonian structure and the quotient of the spectral curve remains unclear.\par\smallskip
It is obvious that the only condition of Lemma \ref{lemma11} which needs to be checked is that all singular points of the spectral curve belong to the set $\mathrm{Fix}\, \hat \tau$ where $\hat \tau$ is given by $(\lambda, \mu) \to (-\bar \lambda, -\bar \mu)$. In terms of the parabolic diagram, this condition means that all intersections are real and belong to the set $\{ x \in \R, y > 0\} \cup \{y = \infty\}$. To show that this condition implies stability, we need to describe the gauge group orbit of $L_\lambda^0 = M_0 + \lambda J^2$.
The gauge group consists of diagonal unitary matrices, so the $G$-orbit of $M_0$ in this case is not discrete. However, its intersection with the phase space, i.e. with the set of real skew-symmetric matrices, is finite and consists of matrices
	\begin{align*}
\left(\begin{array}{ccccccc}0 & \pm m_1 &  &  &  &  &    \\ \mp m_1 & 0 &  &  &  &  &    \\ &  & \ddots &  &  &  &    \\ &  &  & 0 & \pm m_l &  &    \\ &  &  & \mp m_l & 0 &  &    \\ &  &  &  &  & 0 &    \\ &  &  &  &  &  & \ddots  \end{array}\right).
\end{align*}
%As it is easy to see, if this condition is satisfied, then Theorem \ref{thm2} together with Proposition \ref{stab} imply the stability of $M_0$. On the other hand, it is shown in \cite{MBR} that this condition is also a necessary condition for stability, which implies the following result.
We conclude that if all intersections on the parabolic diagram  are either real and belong to the upper half-plane, or infinite, then the equilibrium is stable. On the other hand, it is shown in \cite{nlin} that this condition is also a necessary condition for stability, which implies the following:
\begin{theorem}\label{thm3}
A regular equilibrium of the torque-free multidimensional rigid body is stable if and only if all singular points on the associated parabolic diagram are either real and belong to the upper half-plane, or infinite.
\end{theorem}
A weaker version of this result was proved in \cite{nlin} by means of the bi-Hamiltonian approach. It included an additional requirement that there are no tangency points on the parabolic diagram (i.e. all intersections are of order $1$). 
%Theorem \ref{thm2} allows to omit this condition.
\par
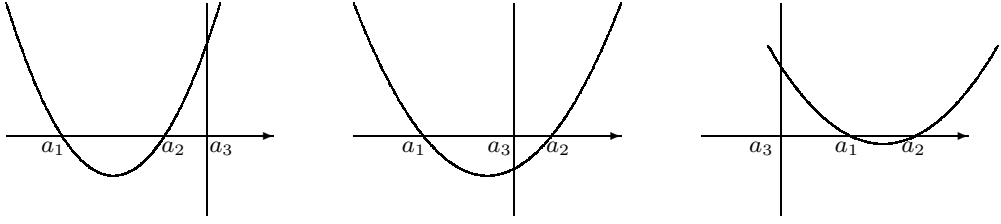
\begin{figure}[t]
{\begin{picture}(500,60)
%\put(0,0){\includegraphics[scale = 0.38]{3d_short.jpg}}
%\put(140,0){\includegraphics[scale = 0.38]{3d_inter.jpg}}
%\put(280,0){\includegraphics[scale = 0.38]{3d_long.jpg}}
%\put(35,23){\tiny{$\lambda_1^2$}}
%\put(61,23){\tiny{$\lambda_2^2$}}
%\put(88,23){\tiny{$\lambda_3^2$}}
\put(45,0){
\qbezier(0,80)(40,-50)(80,80)
\put(0,30){\vector(1,0){100}}
\put(75,0){\line(0,1){80}}
\put(13,24){\small{$a_1$}}
\put(58,24){\small{$a_2$}}
\put(76,24){\small{$a_3$}}
}
\put(175,0){
\qbezier(0,80)(50,-50)(100,80)
\put(0,30){\vector(1,0){100}}
\put(60,0){\line(0,1){80}}
\put(18,24){\small{$a_1$}}
\put(50,24){\small{$a_3$}}
\put(72,24){\small{$a_2$}}
}
\put(305,0){
\qbezier(25,64)(68,-10)(111,64)
\put(0,30){\vector(1,0){100}}
\put(30,0){\line(0,1){80}}
\put(18,24){\small{$a_3$}}
\put(50,24){\small{$a_1$}}
\put(75,24){\small{$a_2$}}
}

\end{picture}}
\caption{Parabolic diagrams for the three-dimensional body. 
%Rotations about the long, middle, and short axis.
}\label{3dpd}
\end{figure}
\begin{figure}[t]
{\begin{picture}(500,100)

\put(20,10){
\qbezier(30,80)(75,-100)(120,80)
\qbezier(20,80)(80,-30)(140,80)
%\qbezier(20,75)(100,-20)(180,75)
\put(10,30){\vector(1,0){120}}
\put(37,24){\small{$a_1$}}
\put(63,32){\small{$a_3$}}
\put(90,32){\small{$a_4$}}
\put(105,24){\small{$a_2$}}
}

\put(150,10){
\qbezier(30,80)(75,-100)(120,80)
\qbezier(46,80)(83,-60)(120,80)
%\qbezier(20,75)(100,-20)(180,75)
\put(10,30){\vector(1,0){120}}
\put(37,24){\small{$a_1$}}
\put(63,32){\small{$a_3$}}
\put(94,32){\small{$a_4$}}
\put(105,24){\small{$a_2$}}

}
\put(280,10){
\qbezier(30,80)(75,-100)(120,80)
\qbezier(52,80)(83,-80)(114,80)
%\qbezier(20,75)(100,-20)(180,75)
\put(10,30){\vector(1,0){120}}
\put(37,24){\small{$a_1$}}
\put(65,32){\small{$a_3$}}
\put(93,32){\small{$a_4$}}
\put(105,24){\small{$a_2$}}
}

\end{picture}}
\caption{Stability loss under Hamiltonian Hopf bifurcation for the four-dimensional rigid body.}\label{4dpd}
\end{figure}
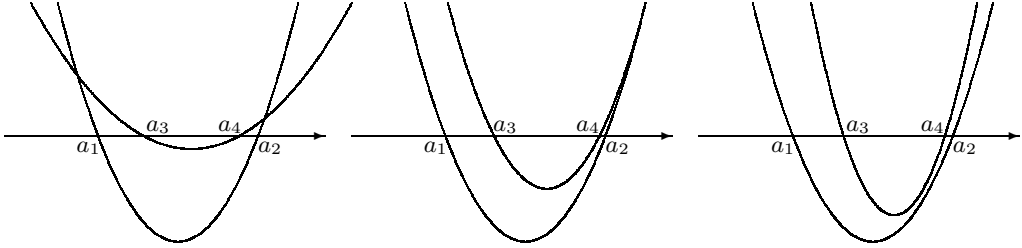
Parabolic diagrams for the three-dimensional body are depicted in Figure \ref{3dpd}. Classical results on stability are immediately recovered. Much more examples of parabolic diagrams can be found in \cite{nlin}.
%\begin{remark}
%The result obtained is actually true for a large class of integrable systems called \textit{Manakov tops}. These systems were defined in \cite{Manakov} as a generalization of (\ref{MRB}).
%\end{remark}
\begin{remark}
Parabolic diagrams were introduced in \cite{JGP} to describe the spectrum of the bi-Hamiltonian structure related to the equation (\ref{MRB}). In the context of spectral curves, the parabolic diagram is the quotient of the spectral curve by the holomorphic involution $(\lambda, \mu) \to (-\lambda, -\mu)$. The relation between the bi-Hamiltonian structure and the quotient of the spectral curve remains unclear.\par\smallskip
\end{remark}
\begin{remark}
	The stability problem for the four-dimensional rigid body was solved almost completely in \cite{Marshall, Ratiu}. However, there is one case in which the instability conclusions of these papers seem to be incorrect. This is case IIIb of (3.13)  in Feh{\'e}r and Marshall \cite{Marshall}, and case V in Theorem 5.3 of Birtea et al. \cite{Ratiu}. It is the case with a tangency point in the upper half-plane depicted in the second diagram in Figure \ref{4dpd}. Theorem \ref{thm3} implies that this equilibrium is stable. This point lies on the boundary of the set of stable equilibria and corresponds to the Hamiltonian Hopf bifurcation (see Section \ref{LSect}). As the ratio $m_1^2 : m_2^2$ grows, a stable regime (first diagram in Figure \ref{4dpd}) is being replaced by an unstable regime (third diagram in Figure \ref{4dpd}). Note that this bifurcation is only possible if $[a_1,a_2] \supset [a_3,a_4]$, or $[a_1,a_2] \subset [a_3,a_4]$. If $[a_1,a_2] \cap [a_3,a_4]$ is empty, then the rotation is stable. And if it is a proper subset of both $[a_1,a_2]$ and $[a_3,a_4]$, then the rotation is unstable.

\end{remark}
\begin{remark}
\label{exotic}
There also exist equilibria of the equation (\ref{MRB}) which do not satisfy the regularity condition - the so-called \textit{exotic equilibria} \cite{NRE}.  The difference between regular and exotic equilibria is the following. The equation \eqref{MRB} is completely integrable can be included into an hierarchy of commuting flows.
%$$
%\dot M = [M, \mathrm{res}_{\lambda = 0}(\lambda^{-j}(M+\lambda J^2)^k)]
%$$
Regular equilibria are those which are common for all these flows \cite{biham}, while exotic equilibria are fixed points of the particular system  \eqref{MRB}. Therefore, the technique of the present paper cannot be applied to exotic equilibria (see Remark  \ref{common}). However, if we are able to prove that \eqref{MRB} is non-resonant, then we can assert that all exotic equilibria are unstable (Remark \ref{common}). Non-resonance was proved in \cite{PhD}, however some technical details in the proof are still missing. A revised proof will be published elsewhere.
%Taking into account Remark \ref{common}, we can assert that exotic equilbria

%It was claimed in the author's thesis \cite{PhD} that these equilibria are unstable, however some technical details in the proof are still missing. A revised proof will be published elsewhere.\par

%Note that regular equilibria are characterized by the fact that they are common equilibria, while exotic equilibria are fixed points of the particular system (\ref{MRB}).
\end{remark}

\section{Several counterexamples}\label{ceSect}
\begin{cexample}\label{ce1}
Let us show that Condition 1 of Lemma \ref{lemma11} cannot be avoided.
Let $\tau$ be the involution $\lambda \to -\bar \lambda$, and let $\sigma = -1$. Let $J^2 = \mathrm{diag}(a_1, a_2, a_3)$ where $a_1 < a_2 <a_3$ are real, positive, and distinct.
Then the space $\widetilde{\mathcal P}_{1,3}(J^2, \tau, \sigma)$ consists of Lax matrices of the form $M+\lambda J^2$ where $M \in \un(3)$ is skew-Hermitian. Let
$$
M_0 = \left(\begin{array}{ccc}0 & 0 & 1 \\0 & 0 & 0 \\-1 & 0 & 0\end{array}\right)
$$
Then the spectral curve corresponding to $M_0 + \lambda J^2$ is given by
\begin{align*}
P(\lambda,\mu) =\left(1 + ( a_{1}\lambda- \mu)(a_3\lambda - \mu)\right)(a_2\lambda - \mu) = 0.
\end{align*}
This curve has a singular point $\lambda_0 = ((a_3 - a_2)(a_2-a_1))^{-1/2}, \mu_0 = a_2\lambda_0$ which does not belong to $\mathrm{Fix}\,\widehat \tau$, while all other assumptions of Lemma \ref{lemma11} are satisfied. Let us show that the statement of Lemma \ref{lemma11} does not hold. The gauge group orbit of $M_0+\lambda J^2$ is
\begin{align*}
%\label{ggo}
\left\{\left.  \left(\begin{array}{ccc}0 & 0 & y \\0 & 0 & 0 \\- \bar y & 0 & 0\end{array}\right)+\lambda J^2 \,\right\vert\, |y| = 1 \right\}.
\end{align*}
%where $\phi$ is real, and $i^2 = -1$.
The isospectral variety is
\begin{align*}
%\label{isv}
\left\{\left.\left(\begin{array}{ccc}0 & z & y \\-\bar z & 0 & x \\-\bar y & -\bar x & 0\end{array}\right)+\lambda J^2  \,\right\vert\, 
\begin{aligned}
a_1|x|^2 + a_2|&y|^2 + a_3|z|^2 = a_2, \\ |x|^2 + |&y|^2 + |z|^2 = 1,  \\ x&\bar y z \in \R \end{aligned}\right\}.
\end{align*}
It is clear that these two sets are distinct. In particular, the real part of the isospectral variety is the intersection of an ellipsoid with a concentric sphere radius equal to the middle semi-axis of the ellipsoid, i.e. two circles intersecting at two points, while the real part of the gauge group orbit is two points $(0,\pm 1, 0)$.\par
The statement of Theorem \ref{thm1} does not hold as well. Let us take $$\mathcal L = \{M+\lambda J^2 \mid M\in \so(3, \R)\}$$ and consider the Lax equation \eqref{laxMRB}, which is the Euler equation of the free three-dimensional rigid body. Then $M_0$ is the rotation about middle axis of inertia, which is known to be unstable.
\end{cexample}
\begin{cexample} Let us show that the condition on the genus is also important. Consider $\widetilde{\mathcal P}_{1,3}(J^2, \tau, \sigma)$ from the previous example, and let
$$
M_0 = \left(\begin{array}{ccc}0 & z_0 & y_0 \\- z_0 & 0 & x_0 \\- y_0 & - x_0 & 0\end{array}\right)
$$
where $x_0,y_0,z_0$ are non-zero real numbers. All conditions of Lemma \ref{lemma11} are satisfied except the condition on the genus: the spectral curve is non-singular and has genus $1$. As in the previous example, take $$\mathcal L = \{M+\lambda J^2 \mid M\in \so(3, \R)\}$$ and consider the Lax equation \eqref{laxMRB}. Then, though the second condition of Theorem \ref{thm1} is satisfied, the conclusion of the theorem does not hold: $M_0$ is not a fixed point at all. Consequently, the statement of Lemma \ref{lemma11} does not hold as well. It is actually easy to see that the quotient of the isospectral variety of $M_0$ by the action of the gauge group is a circle.

\end{cexample}
\begin{cexample}
Finally, let us show that Condition 2 of Theorem 2 cannot be avoided as well. Consider $\widetilde{\mathcal P}_{1,3}(J^2, \tau, \sigma)$ from previous examples, and let $A = \mathrm{diag}(ai, bi, ci)$ be a diagonal skew-Hermitian matrix. Consider the equation
\begin{align}\label{gaugeFlow}
 \diffXp{t} (M+\lambda J^2) = [ M+\lambda J^2, A]
\end{align}
which defines a flow on  $\mathcal L =\widetilde{\mathcal P}_{1,3}(J^2, \tau, \sigma)$. Take $$
M_0 = \left(\begin{array}{ccc}0 & 1 & 0 \\-1 & 0 & 0 \\0 & 0 & 0\end{array}\right).
$$
Then Lemma \ref{lemma11} holds true for $M_0 + \lambda J^2$. However, Condition 2 of Theorem \ref{thm1} is not satisfied, and the statement of the theorem is not true: $M_0 + \lambda J^2$ is not a fixed point of \eqref{gaugeFlow}. Note that the flow defined by  \eqref{gaugeFlow} is just the action of a $1$-parametric subgroup of the gauge group.
\end{cexample}
%which means that the vector $(|x|,|y|, |z|)$ lies on the intersection of an ellipsoid 

%\end{example}
\section{Proof of Lemma \ref{lemma11}}\label{proofSec}
We need to prove that it is possible to recover $L_\lambda$ uniquely from its spectral curve up to the action of the gauge group.
Let $\Gamma$ be the normalized compactification of the spectral curve $C$. Since $C$ is reducible, the Riemann surface $\Gamma$ is not connected, but has $k$ connected components $\Gamma_1, \dots, \Gamma_k$. The functions $\lambda, \mu$ are meromorphic on $\Gamma$ and satisfy the equation
$
\det( L_\lambda - \mu E) = 0.
$
Let also
 $\pi \colon \Gamma \to \CP^1 \times \CP^1$ be the mapping $\gamma\mapsto (\lambda(\gamma),\mu(\gamma))$. The image of $\pi$ is the closure of the curve $C$.
\begin{statement}
The meromorphic function $\lambda$ has exactly $n$ poles $\infty_1, \dots, \infty_n$ in $\Gamma$. 
%These points are also the poles of $\mu$, except for the case when $\det  L_\lambda = 0$ for each $\lambda$. In this latter case, $\mu = 0$ at one of the points $\infty_i$.
\end{statement}
\begin{proof}
This follows easily from the simplicity of the spectrum of $ B_m$.
\end{proof}
Let $\Gamma_{0} = \{\gamma \in \Gamma \mid \lambda(\gamma) < \infty, \,\dim \Ker ( L_{\lambda(\gamma)} - \mu(\gamma) \E) = 1\}.$
Consider the mapping $\psi\colon \Gamma_{0} \to \CP^{n-1}$ which maps $\gamma$ to the eigenvector of $ L_{\lambda(\gamma) }- \mu(\gamma) \E$. This mapping defines a holomorphic line bundle over $\Gamma_0$ which is called the \textit{eigenvector bundle}. However, we will consider $\psi$ as a mapping, not as a line bundle. 
The following statement is well-known (see e.g. Audin \cite{Audin}).
\begin{statement}\label{lemma1}
	The mapping $\psi$ can be uniquely extended to a holomorphic mapping $\psi \colon \Gamma  \to \CP^{n-1}$.
\end{statement}
\begin{proof}
The complement $\Gamma \setminus \Gamma_0$ consists of the finite number of points $\gamma_1, \dots, \gamma_k$. Prove that $\psi$ can be extended to $\gamma_i$. Let $z$ be a local coordinate in the neighborhood of $\gamma_i$ such that $z(\gamma_i) = 0$. Without loss of generality, the first row of the comatrix of $ L_{\lambda(z)} - \mu(z)\E$ is non-zero and finite in the punctured neighborhood of $\gamma_i$. Denote this vector by $a(z) = (a^1(z), \dots, a^n(z))$. Clearly, 
\begin{align*}
\psi(z) = (a^1(z): \ldots: a^n(z))
\end{align*}
for $z \neq 0$. At the same time, there exists a positive or negative integer $m$ such that $a(z) = z^mb(z)$ where $b(z) = (b_1(z), \dots, b_n(z)) \neq 0$ at $z = 0$. So, $\psi$ can be extended to $\gamma_i$ by the formula
\begin{align*}
\psi(z) = (b^1(z): \ldots: b^n(z)).
\end{align*}
Uniqueness is obvious.
\end{proof}
\begin{statement}The vectors $\psi(\infty_1), \dots, \psi(\infty_n)$ are the eigenvectors of $ B_m$.\end{statement}
\begin{proof} Apply the continuity argument. \end{proof} 
 Let us normalize $\psi$ to obtain a meromorphic vector-function $h \colon \Gamma \to \Complex^n$ such that $h(\gamma) \neq 0$ for each $\gamma \in\Gamma$:
$$
 h^i = {\psi^i} \left({\sum_{i=1}^n \alpha_i \psi^i}\right)^{-1}.
$$
 The numbers $\alpha_1, \dots, \alpha_n$ should be chosen in such a way that the poles of $h$ are away from the ramification points of $\lambda$, as well from the points which project to singular points of the curve $C$.\par
 Count the number of poles of $h$.
 For generic $\lambda_0 \in \CP^1$, the set $\lambda^{-1}(\lambda_0) = \{ \gamma \in \Gamma \mid \lambda(\gamma) = \lambda\}$ contains exactly $n$ points $\gamma_1 ,\dots, \gamma_{n}$. Following Dubrovin et al. \cite{DKN}, consider the function
\begin{align*}
r(\lambda_0) =\left( \det(h(\gamma_1), \dots, h(\gamma_{n}))\right)^2.
\end{align*}
This expression does not depend on the numeration of points in $\lambda^{-1}(\lambda_0)$, so it is a rational function of $\lambda_0$. Clearly,
\begin{align}\label{divPoles}
(r)_\infty = 2\lambda\left((h)_\infty\right)
\end{align}
where $(f)_\infty$ denotes the divisor of poles of $f$. Therefore, we can count poles of $h$ by counting poles of $r$. To count poles of $r$, we count its zeros. Obviuosly, $r(\lambda_0)$ can only be zero if the spectrum of $L_{\lambda_0}$ is not simple. This may happen in two cases: either if $\lambda_0$ is a branch point of the function $\lambda$, or if the preimage $\lambda^{-1}(\lambda_0)$ contains singular points of $C$, i.e. those points which belong to at least two irreducible components. Below we count the number of zeros of $r$ corresponding to branch points and show that Condition 1 of Lemma \ref{lemma11} guarantees that singular points \textit{do not} contribute to zeros of $r$.
%Further, $r(\lambda)$ is zero if and only if the vectors $h(\gamma_1), \dots, h(\gamma_{n})$ are dependent. This is possible only if the spectrum of $ L_\lambda$ is not simple, i.e. if $\lambda$ is the branch point of $\Gamma$, or if there exist $\gamma_1 \neq \gamma_2 \in \Gamma$ such that $\lambda(\gamma_1) = \lambda(\gamma_2) = \lambda$, and $\mu(\gamma_1) = \mu(\gamma_2)$.\par
%Denote $\pi(\gamma) = (\lambda(\gamma), \mu(\gamma))$.
\begin{statement}\label{imm}
Let $\gamma \in \Gamma \setminus \{\infty_1, \dots, \infty_n\}$, and let $z$ be a local parameter near $\gamma$. Then either $\lambda'_z \neq 0$, or $\mu'_z \neq 0$.
\end{statement}
\begin{proof}
 The set $\Gamma \setminus \{\infty_1, \dots, \infty_n\}$ is just the disjoint union of irreducible components of $C$, so the statement follows from the smoothness of irreducible components.
%This follows from the fact that $\pi(\gamma) \in C$ is either non-singular, or an ordinary multiple point of the spectral curve. See Appendix for details.
\end{proof}
\begin{statement}\label{imm2}
Let $\gamma \in \Gamma \setminus \{\infty_1, \dots, \infty_n\}$, and let $z$ be a local parameter near $\gamma$. Assume that $\lambda'_z(\gamma) = 0$. Then the matrix
 $ L_{\lambda(\gamma)}$ has a non-trivial Jordan block with eigenvalue $\mu(\gamma)$ and generalized eigenvector $h'_z(\gamma)$. 
 %In particular, $h(\gamma)$ and $h'_z(\gamma)$ are linearly independent.
\end{statement}
\begin{proof}
Without loss of generality, assume that $z(\gamma) = 0$.
Differentiating
$$
( L_{\lambda(z)} - \mu(z)\E)h(z) = 0
$$
with respect to $z$ at $z = 0$, obtain
%\begin{align*}
%\left(\diffFX{  L_\lambda}{\lambda}\lambda'(0) - {\mu'}(0)\E\right)h(0) = - ( L_{\lambda(0)} - \mu(0)\E){h'(0)}.
%\end{align*}
\begin{align*}
  ( L_{\lambda(0)} - \mu(0)\E){h'(0)} ={\mu'(0)}h(0).
\end{align*}
By Proposition \ref{imm}, the number $\mu'(0) $ is non-zero, which proves the proposition.
\end{proof}

\begin{statement}\label{lemma3}
Assume that 
$\gamma_{1}, \dots, \gamma_p \in \Gamma$ are distinct, and that $\pi(\gamma_{i}) = \pi(\gamma_{j})$ for any $i,j \in \{1, \dots, p\}$.
%\end{enumerate}
Then the vectors $h(\gamma_1), \dots, h(\gamma_p)$ are linearly independent.
\end{statement}
\begin{proof}
%Assume that 
%$\gamma_{1}, \dots, \gamma_p \in \Gamma$ are distinct,  $\lambda(\gamma_{i}) = \lambda(\gamma_{j}) = \lambda_0$, and $\mu(\gamma_{i}) = \mu(\gamma_{j}) = \mu_0$  for any $i,j \in \{1, \dots, p\}$.
%\end{enumerate}
Let $\lambda(\gamma_i) = \lambda_0$, and $\mu(\gamma_i) = \mu_0$. Then $(\lambda_0, \mu_0) \in C$ is a singular point. Condition 1 of Lemma \ref{lemma11} implies that $L_{\lambda_0}$ is a normal operator (recall that an operator in Hermitian space is called normal if it commutes with its adjoint), therefore it has no non-trivial Jordan blocks. Using Proposition \ref{imm} and Proposition \ref{imm2}, we conclude that $\lambda$ can be taken as a local parameter near $\gamma_1, \dots, \gamma_p$.\par
 Let $\mu = \mu_i(\lambda)$ and $h = h_i(\lambda)$ in the neighborhood of $\gamma_i$.
Assume that
\begin{align}\label{linDep2}
\sum_{i=1}^{p} c_ih_i(\lambda_0) = 0
\end{align}
and prove that $c_i = 0$ for each $i$.
Differentiating the equation
$$
( L_{\lambda} - \mu_i(\lambda)\E)h_i(\lambda) = 0
$$
with respect to $\lambda$ at $\lambda = \lambda_0$, we have
\begin{align}\label{qi}
Q_i = \left({ L'_\lambda}(\lambda_0) - {\mu'_i}(\lambda_0)\right)h_i(\lambda_0) + ( L_{\lambda_0} - \mu_0\E){h'_i(\lambda_0)}{} = 0.
\end{align}
Using (\ref{linDep2}), we obtain
\begin{align*}
 0 = \sum_{i=1}^{p}c_iQ_i = -\sum_{i=1}^{p} c_i{\mu'_i}{}(\lambda_0)h_{i}(\lambda_0) +  ( L_{\lambda_0} - \mu_0\E) \left(\sum_{i=1}^{p} c_i{h'_i(\lambda_0)}{}\right),
\end{align*}
so 
\begin{align}\label{JB}
( L_{\lambda_0} - \mu_0\E) \left(\sum_{i=1}^{p} c_i{h_i(\lambda_0)}{}\right) = \sum_{i=1}^{p} c_i{\mu'_i}{}(\lambda_0)h_{i}(\lambda_0).
\end{align}
Since $L_{\lambda_0}$ has no non-trivial Jordan blocks, (\ref{JB}) implies that
\begin{align*}
\sum_{i=1}^{p} c_i{\mu'_i}(\lambda_0)h_{i}(\lambda_0) = 0.
\end{align*}
Continuing in the same fashion, we obtain 

\begin{align}\label{manyDep2}
\sum_{i=1}^{p} c_i({\mu'_i}(\lambda_0))^kh_{i}(\lambda_0) = 0 \quad \forall\, k \in \mathbb N \cup \{0\}.
\end{align}
Let $I = \{1, \dots, p\}$. Define an equivalence relation on $I$ by
$
i \equiv j \Leftrightarrow \mu'_i(\lambda_0) = \mu'_j(\lambda_0).
$
Denote equivalence classes by $I_1, \dots, I_q$. Then (\ref{manyDep2}) implies that
\begin{align}\label{linDep3}
\sum_{i \in I_s} c_ih_i(\lambda_0) = 0, \quad \forall \, s = 1, \dots, q. 
\end{align}
If all pairwise intersections of irreducible components of $C$ are of order $1$, then each class $I_s$ consists of one element, and we are done.
%or any $s = 1, \dots, p$.\par
Otherwise, differentiate
$$
( L_{\lambda} - \mu_i(\lambda)\E)h_i(\lambda) = 0
$$
two times with respect to $\lambda$ at $\lambda = \lambda_0$:
\begin{align*}
S_i = ( L''_{\lambda_0} - \mu''_i(\lambda_0)\E)h_i(\lambda_0) + 2( L'_{\lambda_0} - \mu'_i(\lambda_0)\E)h'_i(\lambda_0) + ( L_{\lambda_0} - \mu_0\E)h''_i(\lambda_0) = 0. 
\end{align*}
Fixing $s \in \{1, \dots, q\}$ and using (\ref{linDep3}), we get
\begin{align}\label{secDer}
\begin{aligned}
0 = \sum_{i \in I_s} c_i S_i = -&\sum_{i \in I_s} c_i \mu''_i(\lambda_0)h_i(\lambda_0) + 2( L'_{\lambda_0} - \mu'(\lambda_0)\E)\sum_{i \in I_s} c_ih'_i(\lambda_0) + \\ &+  ( L_{\lambda_0} - \mu_0\E) \sum_{i \in I_s} c_i h''_i(\lambda_0)
\end{aligned}
\end{align}
where $\mu'(\lambda_0) = \mu'_i(\lambda_0)$ for $i \in I_s$.\par
Let $ L = L_{\lambda_0} - \mu_0\E$, and let  $L' = L'_{\lambda_0} - \mu'(\lambda_0)\E$. Then it is easy to see that  \begin{align}\label{LPrime} L^* = \alpha L, \mbox{ and }(L')^* = \beta L'\end{align} where $\alpha = \sigma^{-1}$, and $\beta \in \Complex$ is some constant.
%\begin{statement}\label{prod}
%The matrices  $ L_{\lambda_0} - \mu_0\E$ 
%and $ L'_{\lambda_0} - \mu'(\lambda_0)\E$ are skew-Hermitian.
%\end{statement}
%\begin{proof} The matrix $ L_{\lambda_0}$ is skew-Hermitian by definition. Since $ L_\lambda \in \un(n) \otimes \R[\lambda]$, the same is true for $ L'_\lambda$, so $ L'_{\lambda_0}$ is also skew-Hermitian.
%Further, $\mu_0 \in \mathfrak{i}\R$, so $ L_{\lambda_0} - \mu_0\E$ is skew-Hermitian. Moreover, all eigenvalues of $ L_\lambda$ are pure imaginary for any real $\lambda$, which means that $\mu'(\lambda_0) \in \mathfrak{i}\R$, and $ L'_{\lambda_0} - \mu'(\lambda_0)\E$ is also skew-Hermitian. 
%%To prove the second statement, note that $L'_{\lambda_0}$ is skew-Hermitian. 
%%Further, if $ L_\lambda$ then all its eigenvalues are pure imaginary for any real $\lambda$, which means that $\mu'(\lambda_0) \in \mathfrak{i}\R$, and $L'_{\lambda_0} - \mu'(\lambda_0)\E$ is skew-Hermitian. 
%%Ananlogously,  $\mu'(\lambda_0) \in \R$ for type \textbf{B}, so $L'_{\lambda_0} - \mu'(\lambda_0)\E$ is Hermitian.
%% an imaginary number for type \textbf{A}.
%\end{proof}
%Denote $ L_{\lambda_0} - \mu_0\E =  L$, and $  L'_{\lambda_0} - \mu'(\lambda_0)\E =  L'$.
Let $\langle \,,\rangle$ be the standard Hermitian scalar product $\langle x,y\rangle = x^*y$.
Using (\ref{secDer}), we have
\begin{align}\label{muSecDer}
\left\langle\sum_{i \in I_s} c_i \mu''_i(\lambda_0)h_i(\lambda_0), h_j(\lambda_0)\right\rangle = X_{sj} + Y_{sj}
\end{align}
where 
\begin{align*}
\begin{aligned}
 X_{sj}= 2\left\langle L'\left(\sum_{i \in I_s} c_ih'_i(\lambda_0)\right), h_j(\lambda_0) \right\rangle, \quad Y_{sj} =  \left\langle   L\left( \sum_{i \in I_s} c_i h''_i(\lambda_0)\right), h_j(\lambda_0)\right\rangle.\end{aligned}
\end{align*}
Using \eqref{LPrime}, we show that $Y_{sj}$ vanishes:
\begin{align*}
 Y_{sj} =  \alpha\left\langle   \sum_{i \in I_s} c_i h''_i(\lambda_0),  Lh_j(\lambda_0)\right\rangle = 0.
\end{align*}
 Using  \eqref{LPrime} together with (\ref{qi}) and (\ref{linDep3}), we show that $X_{sj}$ is also zero:
\begin{align*}
\begin{aligned}
X_{sj} = 2\beta\left\langle\sum_{i \in I_s} c_ih'_i(\lambda_0),  L'h_j(\lambda_0) \right\rangle  =  -2\beta\left\langle\sum_{i \in I_s} c_ih'_i(\lambda_0),  Lh'_j(\lambda_0) \right\rangle = \\ =
 -2\beta\bar \alpha\left\langle  L\left(\sum_{i \in I_s} c_ih'_i(\lambda_0)\right), h'_j(\lambda_0) \right\rangle  = 2\beta\bar \alpha\left\langle  L'\left(\sum_{i \in I_s} c_ih_i(\lambda_0)\right), h'_j(\lambda_0) \right\rangle = 0.
\end{aligned}
\end{align*}
Using \eqref{muSecDer} we conclude that
\begin{align*}
\left\langle\sum_{i \in I_s} c_i \mu''_i(\lambda_0)h_i(\lambda_0), h_j(\lambda_0)\right\rangle = 0
\end{align*}
for any $s$ and $j$, hence
\begin{align*}
\sum_{i \in I_s} c_i \mu''_i(\lambda_0)h_i(\lambda_0) = 0\quad \forall \, s = 1, \dots, q. 
\end{align*}
Continuing in the same fashion, we obtain
\begin{align}\label{Vander2}
\sum_{i \in I_s} c_i (\mu''_i(\lambda_0))^kh_i(\lambda_0) = 0 \quad \forall \, s = 1, \dots, q, \quad k \in \mathbb N \cup \{0\}.
\end{align}
%  = \\ =
%\pm2\left\langle\sum_{i \in I_s} c_ih'_i(\lambda_0), (L'_{\lambda_0} - \mu'(\lambda_0)\E)h_j \right\rangle &- \left\langle   \sum_{i \in I_s} c_i h''_i(\lambda_0), (L_{\lambda_0} - \mu_0\E)h_j\right\rangle
%\end{align*}
Since all pairwise intersections of irreducible components of $C$ are of order at most two,  we have $\mu''_i(\lambda_0) \neq \mu''_j(\lambda_0)$ for any distinct $i, j \in I_s$. Consequently, (\ref{Vander2}) implies that $c_i = 0$ for each $i$, q.e.d.

\end{proof}

\begin{statement}\label{divZeros}  The divisor of zeros of $r$ is given by
$
(r)_0 = \lambda((\lambda)_R)
$
where $(\lambda)_R$ is the ramification divisor of $\lambda$.
%So the number of zeros of $r$ equals the number of ramification points of $\lambda$ counting multiplicities.
\end{statement}
\begin{proof}
This statement is well-known in the case when the spectral curve is non-singular \cite{DKN}, and as follows from Proposition \ref{lemma3}, singular points do not contribute to $(r)_0$.
	If the set $\lambda^{-1}(\lambda_0)$ contains $n$ distinct points, then Proposition \ref{lemma3} implies that $r(\lambda_0) \neq 0$. So, $r(\lambda_0) = 0$ if and only if $\lambda_0$ is a branch point.
	For simplicity, assume that the set $\lambda^{-1}(\lambda_0) $ contains exactly one simple ramification point of $\lambda$.
	%For simplicity assume that all ramification points of $\lambda$ are simple, and that their $\lambda$ coordinates are pairwise distinct. Consider a ramification point $(\lambda_0, \mu_0) \in \Gamma$.
	Let $z$ be a local coordinate on $\Gamma$ such that $\lambda - \lambda_0 = z^2$.  Then
	\begin{align*}
	\begin{aligned}
	r(\lambda) &= \left( \det( \ldots, h(0) + h'(0)z  + \mathrm{O}(z^2), h(0) - h'(0)z  + \mathrm{O}(z^2), \ldots)\right)^2 =  \\ &= \left( \det( \ldots, 2h'(0)z  +  \mathrm{O}(z^2), h(0) + \mathrm{O}(z), \ldots)\right)^2 = \\ &= (\lambda - \lambda_0)\left( \det( \ldots, 2h'(0)  +  \mathrm{O}(z), h(0) + \mathrm{O}(z), \ldots)\right)^2.
	\end{aligned}
	\end{align*}
	By Proposition \ref{imm2}, the vector $h'(0)$ is a generalized eigenvector for $ L_{\lambda_0}$, so the latter determinant is non-zero, which proves that $r(\lambda)$ has a simple zero at $\lambda = \lambda_0$. 
	The proof in the general case is analogous.
\end{proof}
Since $r$ is a meromorphic function,
$
\deg (r)_0 = \deg (r)_\infty.
$
Using (\ref{divPoles}) and Proposition \ref{divZeros}, we have 
\begin{align*}
\deg(h)_\infty = \frac{1}{2}\deg (\lambda)_R,
\end{align*}
or, by Riemann-Hurwitz formula,
\begin{align*}
 \deg(h)_\infty =  n - k + \sum_{i=1}^{k} g_i
\end{align*}
where $g_i$ is the genus of $\Gamma_i$. If $g_i = 0$ for each $i$, then
\begin{align}\label{poles}
 \deg(h)_\infty =  n - k.
\end{align}
Let $\xi_1, \dots, \xi_n$ be the eigenvectors of $ B$ ordered in such a way that $h(\infty_i) = \xi_i$, and \begin{align*}\infty_1, \dots, \infty_{n_1} \in \Gamma_1, \quad \infty_{n_1+1}, \dots, \infty_{n_1+n_2} \in \Gamma_2,\quad \dots.\end{align*}
Denote by $h^1, \dots, h^n$ the components of $h$ in the basis $\xi_1, \dots, \xi_n$:
$
h =\sum h^i\xi_i.
$ Then $h^1(\infty_2) =  \dots = h^1(\infty_{n_1}) =0,$ but $h^1(\infty_1) \neq 0$.  Consequently, \begin{align*}\deg(h\mid_{\Gamma_1})_\infty \geq \deg(h^1\mid_{\Gamma_1})_\infty = \deg(h^1\mid_{\Gamma_1})_0 \geq n_1 - 1.\end{align*}Analogously,
\begin{align*}
\deg(h\mid_{\Gamma_i})_\infty \geq n_i - 1 \mbox{ for $i = 1, \dots, k$.}
\end{align*}
Comparing with (\ref{poles}), we conclude that
\begin{align}\label{poles2}
\deg(h\mid_{\Gamma_i})_\infty = n_i - 1.
\end{align}
%Denote $m_i = n_1 + \dots + n_i$.
%Take $j > m_i$ or $j \leq m_{i-1}$. Then $h^j$ vanishes at $n_i$ points $\infty_{m_{i-1} +1}, \dots,\infty_{m_i}$. Using (\ref{poles2}), conclude that $h^j \equiv 0$ on $\Gamma_i$ for each $j > m_i$ or $j \leq m_{i-1}$. This just means that $ L_\lambda$ is block-diagonal, which proves Assertion I of the theorem.\par\smallskip
%$$h^j = 0 \mbox{ for } j \notin $$ 
%\begin{enumerate}
%\item $\lambda$ is the branch point of $\Gamma$.
%\item 
%\end{enumerate}
Identify each $\Gamma_i$ with the standard Riemann sphere, and renormalize $h$ in such a way that all coordinates of $h\mid_{\Gamma_i}$ are polynomials of degree $n_i -1$. Each of these coordinates vanishes at exactly $n_i-1$ points which are fixed, so all coordinates of $h$ are defined uniquely modulo a constant factor. This means that $L_\lambda$ can be recovered uniquely up to conjugation by a matrix which commutes with the leading term $B$. Let us show that this matrix can be chosen to be unitary.\par
Let $L_\lambda = Q^{-1}L_\lambda^0 Q$. Then $L_\lambda^* = Q^*(L_\lambda^0)^*(Q^{-1})^*$ which can be rewritten as
\begin{align}
\label{conj1}
 L_{\tau(\lambda)} = Q^*L^0_{\tau(\lambda)}(Q^{-1})^*.
\end{align}
On the other hand,
\begin{align}
\label{conj2}
 L_{\tau(\lambda)} = Q^{-1}L^0_{\tau(\lambda)}Q.
\end{align}
Comparing \eqref{conj1} and \eqref{conj2}, we conclude that $L_\lambda^0$ commutes with $QQ^*$ for any $\lambda$. Let $S = \sqrt{QQ^*}$. Then the polar decomposition of $Q$ is $Q = SU$ where $U$ is unitary. Since $QQ^*$ commutes with $L_\lambda^0$, so does $S$, therefore
$$L_\lambda = U^{-1}L_\lambda^0 U.$$
Obvoiusly, $U$ commutes with the leading term $B$, so it belongs to the gauge group $G$, q.e.d.
%  Therefore, if $L_\lambda$ and $L_\lambda^0$ satisfy the conditions of Lemma \ref{lemma11}, then they are conjugated by a matrix which is diagonal in the eigenbasis of $B_m$, i.e. an element of the gauge group, which proves the lemma.
%This ambiguity is resolved by rescaling the eigenvectors of $ B_m$. As a result, the matrix of the operator $ L_\lambda$ in the eigenbasis for $ B_m$ is the same as the matrix of the operator ${\widetilde L_\lambda}$ in the eigenbasis for ${\widetilde B}_m$. So, $ L_\lambda$ and ${\widetilde L_\lambda}$ are conjugate, q.e.d.
\bibliographystyle{unsrt}
\bibliography{Lax}

\begin{thebibliography}{10}

\bibitem{DKN}
B.A. Dubrovin, I.M. Krichever, and S.P. Novikov.
\newblock Integrable systems. {I}.
\newblock {\em Dynamical systems IV, Encyclopaedia Math. Sci}, 4:177--332,
  1985.

\bibitem{DMN}
B.A Dubrovin, V.B. Matveev, and S.P. Novikov.
\newblock Non-linear equations of {K}orteweg-de {V}ries type, finite-zone
  linear operators, and {A}belian varieties.
\newblock {\em Russian Mathematical Surveys}, 31(1):59, 1976.

\bibitem{Dubrovin}
B.A. Dubrovin.
\newblock Theta functions and non-linear equations.
\newblock {\em Russian Mathematical Surveys}, 36(2):11, 1981.

\bibitem{babelon}
O.~Babelon, D.~Bernard, and M.~Talon.
\newblock {\em Introduction to classical integrable systems}.
\newblock Cambridge University Press, 2003.

\bibitem{Audin}
M.~Audin.
\newblock {\em Spinning Tops: A Course on Integrable Systems}.
\newblock Cambridge Studies in Advanced Mathematics. Cambridge University
  Press, 1999.

\bibitem{audin2}
M.~Audin.
\newblock Hamiltonian monodromy via {P}icard-{L}efschetz theory.
\newblock {\em Communications in Mathematical Physics}, 229(3):459--489, 2002.

\bibitem{Manakov}
S.V. Manakov.
\newblock Note on the integration of {E}uler's equations of the dynamics of an
  n-dimensional rigid body.
\newblock {\em Functional Analysis and Its Applications}, 10:328--329, 1976.

\bibitem{bobenko1989}
A.~I. Bobenko, A.~G. Reyman, and M.~A. Semenov-Tian-Shansky.
\newblock The {K}owalewski top 99 years later: a {L}ax pair, generalizations
  and explicit solutions.
\newblock {\em Communications in Mathematical Physics}, 122(2):321--354, 1989.

\bibitem{Jovanovich}
B.~Jovanovich.
\newblock The {J}acobi-{R}osochatius problem on an ellipsoid: the {L}ax
  representations and billiards.
\newblock {\em Archive for Rational Mechanics and Analysis}, 210(1):101--131,
  2013.

\bibitem{Adler}
M.~Adler and P.~van Moerbeke.
\newblock Linearization of {H}amiltonian systems, {J}acobi varieties and
  representation theory.
\newblock {\em Advances in Mathematics}, 38(3):318--379, 1980.

\bibitem{Adler2}
M.~Adler and P.~Van~Moerbeke.
\newblock Completely integrable systems, {E}uclidean {L}ie algebras, and
  curves.
\newblock {\em Advances in mathematics}, 38(3):267--317, 1980.

\bibitem{MF}
A.S. Mishchenko and A.T. Fomenko.
\newblock Euler equations on finite-dimensional {L}ie groups.
\newblock {\em Mathematics of the USSR-Izvestiya}, 12(2):371--389, 1978.

\bibitem{Marshall}
L.~Feh{\'e}r and I.~Marshall.
\newblock Stability analysis of some integrable {E}uler equations for {${\rm
  SO}(n)$}.
\newblock {\em J. Nonlinear Math. Phys.}, 10(3):304--317, 2003.

\bibitem{Spiegler}
A.~Spiegler.
\newblock {\em Stability of generic equilibria of the 2N dimensional free rigid
  body using the energy-{C}asimir method}.
\newblock PhD thesis, University of Arizona, 2006.

\bibitem{Casu}
I.~Ca\c{s}u.
\newblock On the stability problem for the $\mathfrak{so}(5)$ free rigid body.
\newblock {\em International Journal of Geometric Methods in Modern Physics},
  8:1205--1223, 2011.

\bibitem{Ratiu}
P.~Birtea, I.~Ca\c{s}u, T.~Ratiu, and M.~Turhan.
\newblock Stability of equilibria for the $\mathfrak{so}(4)$ free rigid body.
\newblock {\em Journal of Nonlinear Science}, 22:187Ð212, 2012.

\bibitem{JGP}
A.~Izosimov.
\newblock Stability in bihamiltonian systems and multidimensional rigid body.
\newblock {\em Journal of Geometry and Physics}, 62(12):2414 -- 2423, 2012.

\bibitem{nlin}
A.~Izosimov.
\newblock Stability of relative equilibria of multidimensional rigid body.
\newblock {\em Nonlinearity}, 27(6):1419, 2014.

\bibitem{BolBorStab}
A.V. Bolsinov, A.V. Borisov, and I.S. Mamaev.
\newblock Topology and stability of integrable systems.
\newblock {\em Russian Mathematical Surveys}, 65(2):259--318, 2010.

\bibitem{intsys}
A.V. Bolsinov and A.T. Fomenko.
\newblock {\em Integrable Hamiltonian systems. Geometry, Topology and
  Classification}.
\newblock CRC Press, 2004.

\bibitem{mvm}
P.~Van~Moerbeke and D.~Mumford.
\newblock The spectrum of difference operators and algebraic curves.
\newblock {\em Acta Mathematica}, 143(1):93--154, 1979.

\bibitem{ratiu1982lagrange}
T.~Ratiu and P.~Van~Moerbeke.
\newblock The {L}agrange rigid body motion.
\newblock In {\em Annales de l'institut Fourier}, volume~32, pages 211--234.
  Institut Fourier, 1982.

\bibitem{gavrilov}
L.~Gavrilov and Zhivkov A.
\newblock The complex geometry of {L}agrange top.
\newblock {\em L'Enseignement Mathematique}, 44:133--170, 1998.

\bibitem{Arnold}
V.I. Arnold.
\newblock {\em Mathematical Methods of Classical Mechanics}.
\newblock Springer-Verlag, 1978.

\bibitem{Cushman}
R.~Cushman and J.C. Meer.
\newblock The {H}amiltonian {H}opf bifurcation in the {L}agrange top.
\newblock In C.~Albert, editor, {\em G\'eom\'etrie Symplectique et
  M\'ecanique}, volume 1416 of {\em Lecture Notes in Mathematics}, pages
  26--38. Springer Berlin Heidelberg, 1990.

\bibitem{Konyaev}
A.Yu. Konyaev.
\newblock Bifurcation diagram and the discriminant of a spectral curve of
  integrable systems on {L}ie algebras.
\newblock {\em Sbornik: Mathematics}, 201(9):1273, 2010.

\bibitem{Brailov}
Yu.~A. Brailov.
\newblock Geometry of translations of invariants on semisimple lie algebras.
\newblock {\em Sbornik: Mathematics}, 194(11):1585, 2003.

\bibitem{biham}
A.V. Bolsinov and A.A. Oshemkov.
\newblock Bi-hamiltonian structures and singularities of integrable systems.
\newblock {\em Regular and Chaotic Dynamics}, 14:431--454, 2009.

\bibitem{Rybnikov}
L.G. Rybnikov.
\newblock Centralizers of certain quadratic elements in {P}oisson-{L}ie
  algebras and the method of translation of invariants.
\newblock {\em Russian Mathematical Surveys}, 60(2):367, 2005.

\bibitem{NRE}
A.~Izosimov.
\newblock A note on relative equilibria of a free multidimensional rigid body.
\newblock {\em Journal of Physics A: Mathematical and Theoretical},
  45(32):325203, 2012.

\bibitem{PhD}
A.~Izosimov.
\newblock {\em Singularities of bihamiltonian systems and the multidimensional
  rigid body}.
\newblock PhD thesis, Loughborough University, 2012.

\end{thebibliography}
\end{document}